%% file: main.tex
\documentclass[11pt]{article}

\usepackage[utf8]{inputenc}
\usepackage[a4paper, margin=2cm, top=2cm, bottom =3cm]{geometry}

\usepackage{mathtools, nccmath,amsmath,amssymb,amsfonts,amstext, dsfont, color}
\usepackage{amsthm}

\usepackage{booktabs}

\usepackage{subcaption}
\usepackage{mathtools, stmaryrd}
\usepackage{xparse} 
\usepackage{mathtools, stmaryrd}
\usepackage{bm}
\usepackage{algorithm}
\usepackage{algorithmicx}
\usepackage{algpseudocode}
\usepackage{nicefrac}

\usepackage{hyperref}

\newtheorem{lemma}{Lemma}[section]
\newtheorem{proposition}{Proposition}[section]
\newtheorem{theorem}{Theorem}[section]

\newtheorem{assumption}{Assumption}[section]
\newtheorem{conjecture}{Conjecture}[section]

\newcommand{\Proba}[1]{\mathrm{Pr}\left(#1\right)}
\newcommand{\ProbaS}[2]{\mathrm{Pr}\left(#1\given #2\right)}
\newcommand{\E}[1]{\mathbb{E}\left[#1\right]}

\newcommand\givenbase[1][]{\:#1\lvert\:}

\newcommand{\CMS}{\texttt{CMS}} 
\newcommand{\CMSCU}{\texttt{CMS-CU}}
\newcommand{\UBCU}{\texttt{CU-UB}}
\newcommand{\LBCU}{\texttt{CU-LB}}

\newcommand{\CBFCU}[1]{\mathbf{Y}(#1)} 
\newcommand{\CBFCUElement}[2]{Y_{#1}(#2)}

\newcommand{\Excess}{\boldsymbol{\Delta}}

\newcommand{\LBCBFCU}[2]{\mathbf{Y}^{\text{lb},#1}(#2)}
\newcommand{\LBCBFCUElement}[3]{ Y^{\text{lb},#1}_{#2}(#3)}
\newcommand{\LBErrorstar}[2]{e^{*,\text{lb}, #1}(#2)}
\newcommand{\LBGap}[1]{G^{\text{lb},#1}}
\newcommand{\LBMatrixTrans}[1]{\mathbf{P}^{\text{lb},#1}}
\newcommand{\LBProbaTrans}[3]{p^{\text{lb},#1}_{\mathbf{#2},\mathbf{#3}}}
\newcommand{\LBMatrixError}[1]{\mathbf{B}^{\text{lb},#1}}
\newcommand{\LBProbaError}[3]{\beta^{\text{lb},#1}_{\mathbf{#2},\mathbf{#3}}}
\newcommand{\LBOccupancy}[1]{\mathbf{\boldsymbol{\pi}}^{\text{lb},#1}}
\newcommand{\LBOccupancyEle}[2]{\pi^{\text{lb},#1}_{\mathbf{#2}}}
\newcommand{\LBExcess}[1]{\boldsymbol{\Delta}^{\text{lb},#1}}
\newcommand{\LBExcessEle}[2]{\Delta^{\text{lb},#1}_{#2}}
\newcommand{\LBVMin}[1]{V^{\text{lb},#1}}
\newcommand{\LBNumberVMin}[1]{C^{\text{lb},#1}}
\newcommand{\LBGamma}{\boldsymbol{\Gamma}^{\text{lb},g}}
\newcommand{\LBGammaEle}{\Gamma^{\text{lb},g}_{l}}

\newcommand{\UBCBFCU}[2]{\mathbf{Y}^{\text{ub},#1}(#2)}
\newcommand{\UBCBFCUElement}[3]{Y^{\text{ub},#1}_{#2}(#3)}
\newcommand{\UBErrorstar}[2]{e^{*,\text{ub}, #1}(#2)}

\newcommand{\UBMatrixTrans}[1]{\mathbf{P}^{\text{ub},#1}}

\newcommand{\UBMatrixError}[1]{\mathbf{B}^{\text{ub},#1}}

\newcommand{\UBExcess}[1]{\boldsymbol{\Delta}^{\text{ub},#1}}

\newcommand{\UBVMin}[1]{V^{\text{ub},#1}}
\newcommand{\UBNumberVMin}[1]{C^{\text{ub},#1}}
\newcommand{\UBGamma}{\boldsymbol{\Gamma}^{\text{ub},g}}
\newcommand{\UBGammaEle}{\Gamma^{\text{ub},g}_{l}}

\algblockdefx[CustomWhile]{While}{EndWhile}[1]{\textbf{While} #1 \textbf{do}}{\textbf{End While}}

\newcommand{\aIntInterval}[1]{\llbracket #1 \rrbracket}
\newcommand{\abIntInterval}[2]{\llbracket #1,  #2 \rrbracket}
\let\given\givenbase

\DeclarePairedDelimiterX\Basics[1](){\let\given\sgiven #1}

\DeclareMathOperator*{\argmin}{arg\,min}

\title{Count-Min Sketch with Conservative Updates: Worst-Case Analysis}

\author{Younes Ben Mazziane \\
        Centre Inria d'Université Côte d'Azur, France \\
        \texttt{younesbenmazziane@gmail.com}
        \and
        Othmane Marfoq \\
        Meta, United States \\
        \texttt{omarfoq@meta.com}
}

\date{}

\begin{document}
\maketitle

\begin{abstract}

Count-Min Sketch with Conservative Updates (\texttt{CMS-CU}) is a memory-efficient hash-based data structure used to estimate the occurrences of items within a data stream. \texttt{CMS-CU} stores~$m$ counters and employs~$d$ hash functions to map items to these counters. We first argue that the estimation error in \texttt{CMS-CU} is maximal when each item appears at most once in the stream. Next, we study \texttt{CMS-CU} in this setting. Precisely,  
\begin{enumerate}

    \item In the case where~$d=m-1$, we prove that the average estimation error and the average counter rate converge almost surely to~$\frac{1}{2}$, contrasting with the vanilla Count-Min Sketch, where the average counter rate is equal to~$\frac{m-1}{m}$. 

    \item For any given~$m$ and~$d$, we prove novel lower and upper bounds on the average estimation error, incorporating a positive integer parameter~$g$. Larger values of this parameter improve the accuracy of the bounds. Moreover, the computation of each bound involves examining an ergodic Markov process with a state space of size~$\binom{m+g-d}{g}$ and a sparse transition probabilities matrix containing~$\mathcal{O}(m\binom{m+g-d}{g})$ non-zero entries. 
    
    \item For~$d=m-1$, $g=1$, and as $m\to \infty$, we show that the lower and upper bounds coincide. In general, our bounds exhibit high accuracy for small values of $g$, as shown by numerical computation. For example, for $m=50$, $d=4$, and $g=5$, the difference between the lower and upper bounds is smaller than~$10^{-4}$. 

\end{enumerate}




\end{abstract}

\section{Introduction}
\label{sec:introduction}
\input{Introductions/introduction}

\section{Notation and Problem Formulation}
\label{sec:problem}
\input{Problems/problem}

\section{Related Work}
\label{sec:related}

\input{related}

 \section{Uniform Counters Selection Assumption}
   \label{sec:motivation}
    \input{main_results/motivation}

\section{Exact Computation}
\label{ss:ExactComputation}
\input{main_results/ExactComputation}

\section{Lower and Upper Bounds}
\label{ss:LowerAndUpperBounds}
\input{main_results/LowerandUpperBounds}




\section{Conclusion}
\label{sec:conclusion}

\input{conclusion}

\bibliography{ref}

\appendix

\input{proofs}

\end{document}

%% file: Introductions/introduction.tex
Counting occurrences and determining the most prevalent items within a data stream represents a fundamental challenge~\cite{Basat2021Salsa, CHARIKAR2004CountSketch, cormode2005improved}. Although the concept of counting may appear straightforward initially---simply incrementing a counter upon encountering an item---the reality is far more complex, especially when confronted with the influx of real-time data streams characteristic of many practical applications. As data volumes surge, maintaining a dedicated counter for each potential item becomes unfeasible~\cite{nelson2012sketching}.



Given the impracticality of exact counting, many applications turn to approximate counting algorithms,
to handle massive data streams efficiently, sacrificing perfect accuracy for scalability. \emph{Count-Min Sketch} ($\CMS$) is a standout among these algorithms and is widely adopted for frequency estimation due to its simplicity and effectiveness~\cite{NLPSketchesApp2012,AnomalyDetectionCMS2021,CMSApplicationEngeneiring}. Indeed, its implementation in the popular data processing library DataSketches underscores its importance in the field~\cite{cormode2020smallBook}. 
Moreover, $\CMS$ is a foundational building block of many approximate counting algorithms that augment its capabilities with additional techniques~\cite{DiamondSketch2019, HeavyKeeper2019,hsu2019learningCM}.


$\CMS$~\cite{cormode2005improved} stores a matrix of~$m$ counters organized in~$d$ rows, utilizing $d$ hash functions to assign items to a counter within each row.
In the vanilla~$\CMS$ algorithm, updates involve incrementing the~$d$ counters corresponding to the hashed values of each input item from the stream. When queried for an item's count, the sketch performs lookups by hashing elements and retrieving the minimum value among the counters.
Due to potential collisions, $\CMS$ overestimates items' frequencies.

When increments are positive, a simple modification to~$\CMS$ leads to better accuracy in frequency estimation. This modification, known as \textit{conservative update}~\cite{ConservativeUpdate2002}, avoids unnecessary increments by first comparing stored counts for an item and then only incrementing the smallest ones. $\CMS$ with conservative updates ($\CMSCU$) finds practical applications in caching~\cite{TinyLFU2017}, natural language processing~\cite{NLPSketchesApp2012}, and bioinformatics~\cite{CMSCUApplicationBiology2022}.



The tradeoff between memory cost and accuracy is well-understood in the vanilla $\CMS$ across various data stream assumptions \cite{cormode2005improved, CormodeZipfErrorCMS2005}, quantified by probabilistic bounds on the estimation error. Even though the conservative update constitutes a minor tweak to~$\CMS$, previous studies fall short of fully grasping its advantages. Analytical challenges arise from the fact that, under the conservative update rule, the update of a counter selected by a hash function is contingent upon the status of the remaining chosen counters; essentially, counters do not increment independently but rather in an interrelated manner. We provide an overview of previous findings concerning $\CMSCU$ in Section~\ref{ss:previousResults} before presenting our contributions in Section~\ref{ss:Contributions}.


\subsection{Previous results}
\label{ss:previousResults}

Few works~\cite{bianchi2012modelingCMSCU,GILBloomCMSCU,MAZZIANE2022CMSCU,fusy2023phase,FusyExpCMSCU2023} explored the advantages of the conservative variant of $\CMS$. Among these, \cite{GILBloomCMSCU, MAZZIANE2022CMSCU,fusy2023phase,FusyExpCMSCU2023} assume that the data stream is a stochastic stationary process and propose probabilistic bounds on the error in terms of items' probabilities of appearing. However, the proposed bounds do not capture the benefits of conservative updates for infrequent items. Quantifying the error for such items is crucial, as empirical findings suggest that they experience larger errors compared to items with higher probabilities~\cite{bianchi2012modelingCMSCU,MAZZIANE2022CMSCU,FusyExpCMSCU2023}.

In contrast to~\cite{GILBloomCMSCU,MAZZIANE2022CMSCU,fusy2023phase}, Bianchi et al.~\cite{bianchi2012modelingCMSCU} examine $\CMSCU$ within a particular framework, shown through experimentation to align with a worst-case error scenario. In this setup, at each time step, $d$ distinct counters are selected uniformly at random, rather than being chosen based on the hash values of items in the stream. Bianchi et al.~\cite{bianchi2012modelingCMSCU} approximate the dynamics of this instance of~$\CMSCU$ via a \textit{fluid limit}. However, this approximation lacks theoretical justification and requires recursively solving a system of ordinary differential equations. Further details about previous studies are provided in Section~\ref{sec:related}.



\subsection{Contributions}
\label{ss:Contributions}

Among previous studies of~$\CMSCU$, Bianchi et al.~\cite{bianchi2012modelingCMSCU} are the only ones who attempted to estimate the worst-case error. Both their work and ours examine the scenario where, at each step, $d$ distinct counters are selected uniformly at random, which we call $\CMSCU$ with \textit{uniform counters selection}.

Bianchi et al.~\cite{bianchi2012modelingCMSCU} motivated that $\CMSCU$ with uniform counters selection represents a worst-case error configuration through simulations. In contrast, our work demonstrates that this instance of~$\CMSCU$ is equivalent to a stream where each item appears at most once, assuming the use of ideal hash functions. We then provide a rationale for considering this configuration as a worst-case error regime.



Driven by the understanding that the uniform counters selection setting constitutes a worst-case error scenario, we tackle the problem of computing the average estimation error within this specific regime. This problem is challenging due to the fast-growing state space of the counters' values. We overcome the state space difficulty by resorting to an alternate representation. Our approach enables us to establish, particularly when $d=m-1$, two key results. First, the average estimation error and the average counter rate converge to~$\nicefrac{1}{2}$ almost surely. Second, the \textit{counters gap}, i.e., the difference between the values of the maximal and minimal counter, is smaller than~$1$ with high probability. These results are formally presented in Theorem~\ref{th:main-particular}.

The computation of the exact estimation error becomes more challenging when dealing with arbitrary values of~$m$ and~$d$. In response, we establish novel lower and upper bounds, parametrized by a positive integer~$g$, on the average estimation error of~$\CMSCU$ with uniform counters selection. These bounds become tighter for larger values of~$g$. Moreover, they correspond to the average estimation errors in two variants of~$\CMSCU$ ($\LBCU$ in Alg.~\ref{alg:LB-CBF-CU} and $\UBCU$ in Alg.~\ref{alg:UB-CBF-CU}), where the \textit{counters gap} is at most~$g$. The limited gap property permits to characterize both the asymptotic and finite-time properties of the error by studying the dynamics of Markov chains with a state space of size $\binom{m+g-d}{m}$ and a sparse transition probabilities matrix with $\mathcal{O}\left(m\binom{m+g-d}{m}\right)$ non-zero entries. Theorem~\ref{th:main} details the characteristics of our bounds, and Lemma~\ref{lem:lInftyUinfty} shows that the lower and upper bounds coincide, even for $g=1$ when $d=m-1$ and as $m\to \infty$. Furthermore, Table~\ref{tab:numericalResults} highlights that the accuracy of our bounds with small values of~$g$ extends beyond the particular regime where~$d=m-1$.

The structure of the paper unfolds as follows: Section~\ref{sec:problem} introduces the notation and elucidates the dynamics of $\CMSCU$. More details about previous work are provided in Section~\ref{sec:related}. Section~\ref{sec:motivation} discusses the assumption of uniform counters selection. Our findings for the instances where $d=m-1$ are presented in Section~\ref{ss:ExactComputation}. We lay out our lower and upper bounds and illustrate their accuracy in Section~\ref{ss:LowerAndUpperBounds}. We conclude the paper in Section~\ref{sec:conclusion}. 





%% file: Problems/problem.tex
Consider a dynamically changing vector~$\bm{n}(t)= (n_i(t))_{i\in \mathcal{I}}$, where~$\mathcal{I}$ is a countable set. This vector~$\mathbf{n}$ undergoes updates based on a stream of $T$ items in $\mathcal{I}$, denoted as $\bm{r}=(r_{t})_{t\in \aIntInterval{T}}$, where $\aIntInterval{k}$ represents the set $\{1, \ldots, k\}$ for any positive integer~$k$. At each time step $t$, the value~$n_{r_t}$ is incremented by~$1$. Therefore, the vector~$\bm{n}(t)$ counts items' number of appearances in the stream $\bm{r}$. 

To estimate counts, we focus on a variant of Count-Min Sketch with Conservative Updates ($\CMSCU$) where hash functions map all items to the same array, known as \textit{Spectral Bloom Filter}~\cite{SpectralBloomFilters2003}. Although the distinction between the two data structures is marginal, $\CMSCU$ is more commonly referenced in the literature; hence, we continue to employ the term $\CMSCU$.


$\CMSCU$ dynamically updates an array of~$m$ counters and employs~$d$ hash functions $h_{l}:\mathcal{I}\mapsto \aIntInterval{m},~l\in\aIntInterval{d}$. These hash functions map each item~$i$ to the counters in the set~$\{h_l(i),\; l\in \aIntInterval{d}\}$, denoted as $\mathbf{h}(i)$. Let~$\bm{Y}(t)=(Y_{c}(t))_{c\in \aIntInterval{m}}\in \mathbb{N}^{m}$ be the counters values at step~$t$.
Initially, all counters are set to~$0$, i.e., $\bm{Y}(0)=\bm{0}$. At each step~$t\geq 1$, the sketch compares counter values associated with~$r_t$. $\CMSCU$ then only increments by~$1$ minimal value counters within~$\bm{h}(r_t)$. 
Specifically, the counters update in $\CMSCU$ is as follows, 
\begin{align}\label{e:DynamicsCU}
        \CBFCUElement{c}{t+1} =
        \begin{cases} 
                    \CBFCUElement{c}{t} + 1, &\text{if } c \in \argmin_{u\in \bm{h}(r_{t+1})}
                    \CBFCUElement{u}{t}, \\
                    \CBFCUElement{c}{t}, & \text{otherwise.}
        \end{cases}
    \end{align}
$\CMSCU$ answers a query for an item~$i$'s count at step~$t$ with the minimum value overall corresponding counters, denoted as $\hat n_{i}(t)$. Formally, $\hat{n}_{i}(t) \triangleq \min_{u\in \bm{h}(i)} Y_{u}(t)$. The error for an item $i$ is then equal to~$\hat n_{i}(t) - n_{i}(t)$. We use the symbols $\CBFCUElement{\max}{t}$ and $\CBFCUElement{\min}{t}$ to denote the maximum and minimum value counters at time $t$, respectively. The counters gap at step $t$, denoted as $G(t)$, is then equal to $\CBFCUElement{\max}{t}-\CBFCUElement{\min}{t}$.

Our aim is to theoretically quantify the relation between the error and the number of counters and hash functions used by $\CMSCU$ under two specific conditions. The first one, described in Assumption~\ref{assum:IdealHash}, supposes that hash functions map items uniformly and independently over~$\binom{\aIntInterval{m}}{d}$. The second condition, presented in Assumption~\ref{assum:DistinctRequests}, considers a specific type of stream where each item appears at most once.   
 \begin{assumption}[Ideal Hash Function Model]\label{assum:IdealHash}
        The sets~$\left( \bm{h}(i)\right)_{i\in \mathcal{I}}$ are i.i.d. uniform random variables over $d$-sized subsets of~$\aIntInterval{m}$, denoted as~$\binom{\aIntInterval{m}}{d}$.       
\end{assumption}
 \begin{assumption}\label{assum:DistinctRequests}
    The stream~$\mathbf{r}$ consists of distinct items, i.e.,~$r_t\neq r_s$ for every~$t\neq s\in \aIntInterval{T}$.
\end{assumption}
The errors for any two items absent from the stream $\bm{r}$ are identically distributed when Assumption~\ref{assum:IdealHash} holds. Let~$e^{*}(t,\bm{r})$ denote a random variable following their common probability distribution. Furthermore, for any two streams $\bm{r}$ and $\bm{a}$ satisfying Assumption~\ref{assum:DistinctRequests}, $e^{*}(t,\bm{a})$ and $e^{*}(t,\bm{r})$ share identical probability distribution and we define~$e^{*}(t)$ as a random variable with that distribution. We refer to~$e^{*}(t)$ as the \textit{estimation error} and $\nicefrac{e^{*}(t)}{t}$ as the \textit{average estimation error}.

%% file: related.tex

In this section, we provide more details about previous results concerning $\CMSCU$, briefly mentioned in Section~\ref{ss:previousResults}.

Theoretical investigations into $\CMSCU$ likely commenced with the research conducted by Binachi et al.~\cite{bianchi2012modelingCMSCU}. Their primary focus was studying the dynamics of $\CMSCU$ when $d$ distinct counters are selected randomly at each time step. They employed a \textit{fluid limit} model to characterize the growth of counters. However, it's worth noting that this fluid limit approximation is specifically tailored for scenarios where the ratio $\nicefrac{d}{m}$ tends towards~$0$. Moreover, the paper did not offer a theoretical assessment of the precision of this approximation; instead, it relied on Monte Carlo simulations to validate the accuracy.

Subsequent works studied $\CMSCU$ under a stochastic stationary stream, i.e., $(r_t)_{t\in \aIntInterval{T}}$ are i.i.d. categorical random variables over the set of items~$\mathcal{I}$. For instance, Einziger and Friedman~\cite{GILBloomCMSCU} approximated $\CMSCU$'s dynamics with a stack of \textit{Bloom Filters}~\cite{BroderBloomFilter2003}. Unfortunately, their approach may fail to capture the conservative update advantages~\cite{MAZZIANE2022CMSCU}. Ben Mazziane et al.~\cite{MAZZIANE2022CMSCU} derived upper bounds on the estimation error, providing theoretical support that frequent items barely suffer any error in comparison to infrequent ones. However, their study overlooked the benefits of conservative updates for infrequent items and situations where item request probabilities are similar. Recently, Fusy and Kucherov~\cite{fusy2023phase,FusyExpCMSCU2023} considered the particular case where items' probabilities are all equal, and showed that if the ratio between the number of items in $\mathcal{I}$ and the total number of counters~$m$ remains below a specific threshold, the relative error in $\CMSCU$, namely $(\hat n_{i}(T)- n_{i}(T))/n_{i}(T)$, is asymptotically~$o(1)$, distinguishing it from $\CMS$.

%% file: main_results/motivation.tex
In the introduction, we highlighted experimental findings from prior research~\cite{bianchi2012modelingCMSCU,MAZZIANE2022CMSCU,FusyExpCMSCU2023} concerning the error in~$\CMSCU$, demonstrating two key points. First, infrequent items exhibit greater errors compared to others. 
Second, infrequent items' errors seem maximal when~$d$ distinct counters are uniformly selected at random within $\CMSCU$~\cite{bianchi2012modelingCMSCU}. This particular instance of $\CMSCU$ is formalized in~Assumption~\ref{assum:UniformCounters} and can be interpreted as the fusion of Assumptions~\ref{assum:IdealHash} and~\ref{assum:DistinctRequests}.
\begin{assumption}[Uniform Counters Selection]\label{assum:UniformCounters}
    The sets $(\bm{h}(r_t))_{t\in \mathbb{N}}$ are i.i.d. uniform random variables over~$\binom{\aIntInterval{m}}{d}$.
\end{assumption}


Many works adopt ideal hashing models, such as the one in Assumption~\ref{assum:IdealHash}, despite the impractical computational cost of their design~\cite{BroderBloomFilter2003,bianchi2012modelingCMSCU,GILBloomCMSCU,fusy2023phase}. Remarkably, these models often predict the performance of algorithms operating with practical hash functions of weaker notions of independence. A theoretical rationale for this phenomenon is elucidated in~\cite{mitzenmacher2008simple}.


Assumption~\ref{assum:DistinctRequests} deviates from typical data stream models. For example, previous studies of $\CMSCU$~\cite{GILBloomCMSCU,MAZZIANE2022CMSCU,fusy2023phase} assume a random stationary stream and provide probabilistic bounds on the error in terms of items' probabilities of appearance. However, these bounds fall short of capturing the benefits of $\CMSCU$ when items' probabilities are small. One could perceive a stream conforming to Assumption~\ref{assum:DistinctRequests} as a highly probable realization of a stochastic process where all items' probabilities are equal, and the catalog of items is infinite w.r.t. the stream length, i.e., $T \ll |\mathcal{I}|$. This perspective highlights streams following Assumption~\ref{assum:DistinctRequests} as scenarios where existing bounds fail to capture the advantages of~$\CMSCU$.


We conjecture that the expected estimation error for items in $\CMSCU$ reaches its maximum for infrequent items in a stream adhering to Assumption~\ref{assum:DistinctRequests} when the hash functions are ideal. Let $e_{i}(t,\bm{r})$ be the error of item~$i$ at step~$t$ in stream~$\bm{r}$. Conjecture~\ref{conj:maximalError} precisely states our claim. 

\begin{conjecture}\label{conj:maximalError}
Under Assumption~\ref{assum:IdealHash}, for any item $i$ and stream $\bm{r}$, the following inequality holds: $\E{e_{i}(t,\bm{r})} \leq \E{e^{*}(t)}$.
\end{conjecture}
To support Conjecture~\ref{conj:maximalError}, Proposition~\ref{prop:WorstErrorItemNotReq} formally proves that for any stream, the maximal expected estimation error occurs for items with the lowest count. 

\begin{proposition}\label{prop:WorstErrorItemNotReq}
    Under Assumption~\ref{assum:IdealHash}, for any item~$i$ and stream~$\bm{r}$, the following inequality holds: $\E{e_{i}(t,\bm{r})} \leq \E{e^{*}(t,\bm{r})}$.
\end{proposition}
 Proposition~\ref{prop:WorstErrorItemNotReq} (proof in Appendix~\ref{proof:propositionWorstErrorItem}) partially validates the empirical observation that the expected error of an item~$i$ at time~$t$ in $\CMSCU$ tends to increase as~$n_{i}(t)$ decreases~\cite{bianchi2012modelingCMSCU,MAZZIANE2022CMSCU}. Notably, Bianchi et al.~\cite{bianchi2012modelingCMSCU} experimentally observed that frequent items have negligible errors, whereas infrequent ones are estimated with the same value, called \emph{error floor}, which corresponds to $\E{e^{*}(t,\bm{r})}$ in our context. Proposition~\ref{prop:WorstErrorItemNotReq} and the simulations in~\cite{bianchi2012modelingCMSCU}, suggesting that $\E{e^{*}(t,\bm{r})}\leq \E{e^{*}(t)}$, enhance the plausibility of Conjecture~\ref{conj:maximalError}. The remainder of this paper focuses on computing the average estimation error $\nicefrac{e^{*}(t)}{t}$.

%% file: main_results/ExactComputation.tex
One direct approach to compute the average estimation error in~$\CMSCU$ under Assumption~\ref{assum:UniformCounters}, is to recognize that~$(\mathbf{Y}(t))_t$ forms a discrete-time Markov chain, allowing for the computation of the probability that~$\mathbf{Y}(t)$ is in a specific state. However, the state space of this Markov chain grows fast with~$t$. To deal with the growing state space, we present an alternate representation of $(\mathbf{Y}(t))_t$ enabling a more efficient computation of the average estimation error, especially in the particular case where $d=m-1$.


Our method considers the Markov chain representing the number of counters deviating from the minimum value counter by any~$l\in\mathbb{N}$. We denote the state of this Markov chain at step~$t$ as $\boldsymbol{\Delta}(t)= (\Delta_l(t))_{l\in \mathbb{N}}$ where~$\Delta_l(t)$ is formally defined as,
\begin{align}
        \Delta_l(t) \triangleq \sum_{c=1}^{m} \mathds{1}\left(Y_{c}(t) - Y_{\min}(t) = l   \right).
\end{align}

Theorem~\ref{th:main-particular} examines the stochastic evolution of~$(\Excess(t))_t$ when~$d=m-1$, to characterize the average estimation error and the average counter rate, namely~$\nicefrac{\sum_{u=1}^{m} Y_{u}(t)}{tm}$, in Property~\ref{thPro:errordm1}, and the counters gap in Property~\ref{thPro:gapdm1}. 

\begin{theorem}\label{th:main-particular}
 Under Assumption~\ref{assum:UniformCounters} and when~$d=m-1$, $\CMSCU$ satisfy the following asymptotic properties:
\begin{enumerate}
    \item As $T\to \infty$, both the average estimation error~$\nicefrac{e^{*}(T)}{T}$ and the average counter rate converge almost surely to~$\nicefrac{1}{2}$.  \label{thPro:errordm1}
    \item The \emph{counters gap} $G(T)$ is less or equal than~$1$ w.h.p. in~$T$ and~$m$.  \label{thPro:gapdm1}
\end{enumerate}
\end{theorem}

\begin{proof}[Sketch proof]
In the particular case of~$d=m-1$, $(\Excess(t))_t$ can be modeled as a \textit{birth and death} process, allowing for closed-form expression for its limiting distribution. Leveraging this, the proof of Theorem~\ref{th:main-particular} applies~\cite[Prop. 4.3, p. 215]{ross2014introduction} to prove the almost sure convergence properties. For completeness, we present the proposition. 
\begin{proposition}\cite[Prop. 4.3, p. 215]{ross2014introduction}\label{prop:ergodicMarkovMain}
  Let $(X(t))_t$ be an irreducible Markov chain with stationary probabilities $\Pi_f$ for any $f\in\mathbb{N}$, and let $r$ be a bounded function on the state space. Then, 
            \begin{align}
              \lim_{T\to \infty} \frac{1}{T} \sum_{t=1}^{T} r(X_t) = \sum_{f=0}^{\infty} r(f) \Pi_f , \text{ w.p. 1}. 
            \end{align}
\end{proposition}
For instance, to prove the average estimation error result, it suffices to set~$X(t)$ to $(\delta_t, \Excess(t), \Excess(t-1))$ and $r(X(t))$ to $\delta_t$, where $\delta_t = e^{*}(t) - e^{*}(t-1)$.

The detailed proof is presented in Appendix~\ref{proof:th-main-particular}. 
\end{proof}



Theorem~\ref{th:main-particular} highlights a sharp distinction between~$\CMS$ and~$\CMSCU$ under the uniform counters selection regime, particularly when~$d=m-1$. In this scenario, $\CMS$ with~$T$ rounds can be conceptualized as a \textit{balls and bins} process with~$m$ bins and~$T(m-1)$ balls. Consequently, the average counter rate equals~$\nicefrac{(m-1)}{m}$, and the maximum value counter is~$\nicefrac{T(m-1)}{m} + \Theta(\sqrt{T})$ for large~$T$~\cite{berenbrink2000balanced}. In contrast, when both~$T$ and~$m$ are large, $\CMSCU$ exhibits a slower average counter rate of~$\nicefrac{1}{2}$. Moreover, the deviation from the average counter value does not grow with~$T$.

The contrast in deviation from the average between~$\CMS$ and~$\CMSCU$ mirrors that observed between single-choice and $d$-choice variants in balls and bins processes. Interestingly, for the balls and bins variants, this contrast extends beyond the particular case where~$d=m-1$ to the broader setting where~$d\geq 2$. Specifically, in the single choice, the deviation from the average is~$\mathcal{O}(\sqrt{T})$, whereas in the $d$-choice, it is~$\mathcal{O}\left(\nicefrac{\ln (\ln{m})}{\ln{d}}\right)$~\cite{berenbrink2000balanced}.

%% file: main_results/LowerandUpperBounds.tex



For generic values of~$m$ and~$d$, the computation of the average estimation error by examining the infinite state space Markov process~$(\Excess(t))_t$ is more difficult. To manage the infinite state space, we introduce $\LBCU$ (Algorithm~\ref{alg:LB-CBF-CU}) and $\UBCU$ (Algorithm~\ref{alg:UB-CBF-CU}), two variants of $\CMSCU$ in which the counters gap is capped by a constant~$g\in\mathbb{N}$.  


\begin{algorithm}[t]
\caption{$\LBCU$}
\begin{algorithmic}[1]
\State \textbf{Input:} $(r_t)_{t\in \aIntInterval{T}}$, $g$
\State  $\LBCBFCU{g}{0} \gets \mathbf{0}$
\For{$t$ in $\aIntInterval{T}$}
    \If{$G^{\text{lb,g}}(t-1) = g$, and $\min_{u\in \bm{h}(r_t)} \LBCBFCUElement{g}{u}{t-1}=\LBCBFCUElement{g}{\text{max}}{t-1}$}
        \State $\LBCBFCU{g}{t} \gets \LBCBFCU{g}{t-1}$ \label{line:LBNoUpdate}
    \Else
        \State $\LBCBFCU{g}{t} \gets$ \text{CU}$\left(\LBCBFCU{g}{t-1}, \bm{h}(r_t)\right)$ \label{line:LBCuUpdate}
    \EndIf
\EndFor
\end{algorithmic}
\label{alg:LB-CBF-CU}
\end{algorithm}

\begin{algorithm}[t]
\caption{$\UBCU$}
\begin{algorithmic}[1]
\State \textbf{Input:} $(r_t)_{t\in \aIntInterval{T}}$, $g$
\State $\UBCBFCU{g}{0}  \gets \mathbf{0}$
\For{$t$ in $\aIntInterval{T}$}
    \State $\UBCBFCU{g}{t} \gets$ \text{CU}$\left(\UBCBFCU{g}{t-1}, \bm{h}(r_t)\right)$ \label{line:UbCuUpdate}
    \If{$G^{\text{ub,g}}(t-1) = g$, and $\min_{u\in \bm{h}(r_t)} \UBCBFCUElement{g}{u}{t-1}=\UBCBFCUElement{g}{\text{max}}{t-1}$}
        \State $\UBCBFCUElement{g}{c}{t}= \UBCBFCUElement{g}{c}{t-1} + 1, \; \forall c\in \aIntInterval{m}: \; \UBCBFCUElement{g}{c}{t-1}= \UBCBFCUElement{g}{\text{min}}{t-1}$\label{line:UbMinUpdate}
    \EndIf
\EndFor
\end{algorithmic}
\label{alg:UB-CBF-CU}
\end{algorithm}

We extend the notation employed for~$\CMSCU$ in Section~\ref{sec:problem} to $\LBCU$ and $\UBCU$ by including the superscripts~"$\text{lb},g$" and~"$\text{ub},g$", respectively. Initially, all counters are equal to~$0$. At each step~$t\geq 1$, both $\LBCU$ and $\UBCU$ update counter values based on the previous state, distinguishing between two cases:
  \begin{enumerate}
    \item If the gap between maximum and minimum counter values equals~$g$, and all selected counters have values equal to the maximum counter value, $\LBCU$ performs no updates (line~\ref{line:LBNoUpdate}), while in $\UBCU$, counters with values equal to $\UBCBFCUElement{g}{\text{min}}{t-1}$ are incremented (line~\ref{line:UbMinUpdate}), along with the selected counters (line~\ref{line:UbCuUpdate}).
    \item Otherwise, $\LBCU$ and $\UBCU$ update the counters according to the conservative update in~\eqref{e:DynamicsCU} (lines~\ref{line:LBCuUpdate} and~\ref{line:UbCuUpdate}). 
\end{enumerate}


Define~$\ell_g(T)$ and~$U_{g}(T)$ as the expected average estimation errors in $\LBCU$ and $\UBCU$, respectively. Specifically, $\ell_g(T) = \nicefrac{\E{\LBErrorstar{g}{T}}}{T}$ and $U_g(T) = \nicefrac{\E{\UBErrorstar{g}{T}}}{T}$. Theorem~\ref{th:main} shows that $\ell_g(T)$ and $U_g(T)$ match the expected average estimation error of~$\CMSCU$ ($\E{\nicefrac{e^{*}(T)}{T}}$) when $g \geq T$ (Property~\ref{eq:stationary}). Moreover, when $g < T$, it holds that $\ell_g(T)<\ell_{g+1}(T)$ and $U_g(T)>U_{g+1}(T)$, providing multiple bounds on~$\E{\nicefrac{e^{*}(T)}{T}}$ that progressively tighten for larger values of~$g$ (Property~\ref{eq:sandwich_expectation}).

The capped counters gap in the two~$\CMSCU$ variants facilitates the analysis of their estimation errors under Assumption~\ref{assum:UniformCounters}. This characteristic ensures that, for each variant, the Markov chain representing the number of counters deviating from the minimum value counter by any~$l\in\mathbb{N}$ has a finite state space. 
We denote this Markov chain as $(\LBExcess{g}(t))_t$ for $\LBCU$ and $(\UBExcess{g}(t))_t$ for $\UBCU$. Theorem~\ref{th:main} examines the stochastic evolution of~$(\LBExcess{g}(t))_t$ and~$(\UBExcess{g}(t))_t$ to prove that i) The values $\ell_g(T)$ and $U_g(T)$ can be efficiently computed with a time complexity of~\( \mathcal{O}( (T+g) m \binom{m + g - d}{g}) \) (Property~\ref{eq:complexity}) and, ii) The quantity~$\nicefrac{e^{*}(T)}{T}$ concentrates in the interval~$[\ell_{g}(\infty), U_{g}(\infty)]$ with high probability, where $\ell_{g}(\infty)$ and $U_{g}(\infty)$ are the limits of $\ell_g(T)$ and $U_g(T)$ when~$T$ tends to infinity (Property~\ref{eq:sandwich_proba}). The limits $\ell_{g}(\infty)$ and $U_{g}(\infty)$ are determined by the limiting distributions of $(\LBExcess{g}(t))_t$ and $(\UBExcess{g}(t))_t$.

\begin{theorem}\label{th:main}
   Under Assumption~\ref{assum:UniformCounters}, the sequences $(\ell_g(T))_g$ and $(U_g(T))_g$ along with the average estimation error $\nicefrac{e^{*}(T)}{T}$ satisfy the following properties:
    \begin{enumerate}
        \item When \( g \geq T \), \( \ell_g(T) = \E{\nicefrac{e^{*}(T)}{T}} = U_g(T) \). \label{eq:stationary}
        \item For \( g < T \), \(\ell_{g}(T) < \ell_{g+1}(T) \) and \(U_{g}(T) > U_{g+1}(T) \). \label{eq:sandwich_expectation}
        \item  The time complexity for computing \( \ell_g(T) \) and \( U_g(T) \) is  \( \mathcal{O}( (T+g) m\binom{m + g - d}{g}) \). \label{eq:complexity}
        
        \item As \( T \to \infty \), the probability of \( \ell_g(\infty) \leq \nicefrac{e^{*}(T)}{T} \leq U_g(\infty) \) converges to 1. \label{eq:sandwich_proba}
    \end{enumerate}
\end{theorem}
\begin{proof}[Sketch proof of Properties~\ref{eq:stationary} and~\ref{eq:sandwich_expectation} in Th.~\ref{th:main}]

For any stream of length~$T$ and any value $g\geq T$, $\LBCU$, $\UBCU$, and $\CMSCU$ produce identical counter vectors, i.e., $\LBCBFCU{g}{t}=\UBCBFCU{g}{t}=\CBFCU{t}$, for any $t\in \aIntInterval{T}$. Property~\ref{eq:stationary} is an immediate result of this observation. Moreover, the proof of Property~\ref{eq:sandwich_expectation} builds on Lemma~\ref{lem:monotony}. 
\begin{lemma}
    \label{lem:monotony}
    For any stream $(r_t)_{t\geq0}$, $(\LBCBFCU{g}{t})_{g\geq 0}$ (resp. $(\UBCBFCU{g}{t})_{g\geq 0}$) is an element-wise non-decreasing (resp. non-increasing) sequence. 
\end{lemma}

Lemma~\ref{lem:monotony} (proof in Appendix~\ref{proof:monothony}) implies that~$ \LBCBFCU{g}{T} \leq \CBFCU{T} \leq \UBCBFCU{g}{T}$ for any~$g\in \mathbb{N}$. In particular, this inequality is valid for the errors of the items that did not appear in the stream before time~$T$. Moreover, for any~$g<T$, it is easy to find a sequence~$(\bm{h}(r_t))_{t\geq 0}$ such that~$\LBErrorstar{g}{T} < \LBErrorstar{g+1}{T}$ and~$\UBErrorstar{g+1}{T} < \UBErrorstar{g}{T}$, thus~$\E{\LBErrorstar{g}{T}}$ is strictly smaller than~$\E{\LBErrorstar{g+1}{T}}$ and $\E{\UBErrorstar{g+1}{T}}$ is strictly smaller than~$\E{\UBErrorstar{g}{T}}$. It follows that~$(\ell_g(T))_g$ (resp.~$(U_g(T))_g$) is a decreasing (resp. increasing) sequence, proving Property~\ref{eq:sandwich_expectation}.

\end{proof}

\begin{proof}[Sketch proof of Properties~\ref{eq:complexity} and~\ref{eq:sandwich_proba} in Th.~\ref{th:main}]
We apply Proposition~\ref{prop:ergodicMarkovMain} and its finite-time variant, presented in Lemma~\ref{lem:finiteExpectationErgodic}.
\begin{lemma}\label{lem:finiteExpectationErgodic}
  Let $(X(t))_t$ be a Markov chain with a specific initial state, and define $\Pi_f(t) = \mathbb{P}(X(t) = f)$ for any $f\in \mathbb{N}$. Let $r$ be a bounded function on the state space. Then,
            \begin{align}
              \E{\sum_{t=1}^{T} r(X_t)} = \sum_{t=1}^{T} \sum_{f=0}^{\infty} r(f) \Pi_f(t).
            \end{align}
\end{lemma}
The computation of~$\ell_{g}(T)$ sets~$X(t)$ to $(\delta_t^{\text{lb},g}, \LBExcess{g}(t), \LBExcess{g}(t-1))$ and $r(X(t))$ to $\delta_t^{\text{lb},g}$, where $\delta_t^{\text{lb},g} = \LBErrorstar{g}{t} - \LBErrorstar{g}{t-1}$. Let~$\Omega_g$ be the state space of $(\LBExcess{g}(t))_t$, for any $\bm{k},\bm{k^{'}}\in \Omega_g$, we have that,
 \begin{align}\label{e:homenow}
            \Pi_{1,\bm{k^{'}},\bm{k}}(t) = \LBOccupancyEle{g}{k}(t-1) \cdot \LBProbaError{g}{k}{k^{'}} \cdot \LBProbaTrans{g}{k}{k^{'}}. 
\end{align}
where $\LBOccupancyEle{g}{k}(t)$ represents the probability that $\LBExcess{g}(t)=\mathbf{k}$, $\LBProbaTrans{g}{k}{k^{'}}$ is the one-step transition probability of $(\LBExcess{g}(t))_t$, and~$\LBProbaError{g}{k}{k^{'}}$ is the conditional expectation of~$\delta_{t}^{\text{lb},g}$ given that~$(\LBExcess{g}(t), \LBExcess{g}(t-1)) = (\mathbf{k^{'}},\mathbf{k})$. Combining Lemma~\ref{lem:finiteExpectationErgodic} and \eqref{e:homenow}, $\ell_g(T)$ can computed as, 
\begin{align}\label{e:ExpErrorLBg}
            \ell_g(T) = \frac{1}{T}\left\langle  \sum_{t=0}^{T-1}\LBOccupancy{g}(t) ,  \LBMatrixTrans{g} \odot \LBMatrixError{g} \cdot \mathbf{1}^\intercal \right\rangle,
\end{align}
where $\odot$ denotes the element-wise product, $\mathbf{1}^\intercal$ represents a column vector with $|\Omega_g|$ rows, each entry being $1$, $\LBOccupancy{g}(t)$ is a vector with components $\LBOccupancyEle{g}{k}(t)$, and $\LBMatrixTrans{g}$ and $\LBMatrixError{g}$ are matrices with components $\LBProbaTrans{g}{k}{k^{'}}$ and $\LBProbaError{g}{k}{k^{'}}$ for any $\bm{k}$ and $\bm{k^{'}}$ in $\Omega_g$.

We establish that: 1) $|\Omega_g|= \binom{m+g-d}{d}$, 2) $\LBMatrixTrans{g}$ is sparse with $\mathcal{O}(m\binom{m+g-d}{g})$ non-zero entries, 3) $(\LBExcess{g}(t))_t$ has a limiting distribution, and 4) Computing $\LBMatrixTrans{g} \odot \LBMatrixError{g} \cdot \mathbf{1}^\intercal$ has a time complexity of $\mathcal{O}( mg \binom{m+g-d}{g})$. Consequently, computing $\ell_{g}(T)$ using \eqref{e:ExpErrorLBg} incurs a time complexity of $\mathcal{O}( (T+g)m \binom{m+g-d}{g})$. Moreover, Proposition~\ref{prop:ergodicMarkovMain} and Equation~\eqref{e:homenow} imply that,
 \begin{align}\label{e:LBAlmostSureConvergence}
            \Proba{\lim_{T\to +\infty} \nicefrac{\LBErrorstar{g}{T}}{T}  = \left\langle  \LBOccupancy{g}(\infty) ,  \LBMatrixTrans{g} \odot \LBMatrixError{g} \cdot \mathbf{1}^\intercal \right\rangle} = 1.
  \end{align}
Similar results hold for~$\UBCU$ following the same reasoning. The detailed proof is presented in Appendix~\ref{proof:dynamicsExcessLBUBprocesses}.
\end{proof}

\paragraph*{Accuracy of the bounds}
In the particular case where $d = m - 1$ and $g = 1$, both Markov chains $(\LBExcess{g}(t))_t$ and $(\LBExcess{g}(t))_t$ reduce to systems with only two states. Lemma~\ref{lem:lInftyUinfty} leverages this simplification to determine $\ell_1(\infty)$ and $U_{1}(\infty)$.
\begin{lemma}\label{lem:lInftyUinfty}
 When~$d=m-1$,~$\ell_{1}(\infty)$ and~$U_{1}(\infty)$ have closed form formulas:

 \( \ell_{1}(\infty) = \frac{m-1}{2m-1} ,    \; U_{1}(\infty) = \frac{m}{2m-1}. \)
\end{lemma}

Lemma~\ref{lem:lInftyUinfty}, combined with Theorems~\ref{th:main} and~\ref{th:main-particular}, shows that the interval $[\ell_{1}(\infty),U_1(\infty)]$ converges to $\{\nicefrac{1}{2}\}$, as $m\to \infty$, which coincides with the value of~$\nicefrac{e^{*}(t)}{t}$ when $t\to \infty$. Consequently, for large values of~$m$ when $d=m-1$, our lower and upper bounds are accurate even for~$g=1$.

Numerical computation of $\ell_g(T)$ involves computing the matrices~$\LBMatrixError{g}$ and~$\LBMatrixTrans{g}$ (see~\eqref{e:ExpErrorLBg} and~\eqref{e:LBAlmostSureConvergence}). Similarly, for $U_g(T)$, we need to compute the matrices~$\UBMatrixError{g}$ and~$\UBMatrixTrans{g}$. All these matrices' expressions are derived in Appendix~\ref{proof:dynamicsExcessLBUBprocesses}. These matrices exhibit sparsity. For $\LBMatrixError{g}$ and $\LBMatrixTrans{g}$, the non-zero entries correspond to transitions from state~$\mathbf{k}$ to state~$\mathbf{k{'}}$ only if~$\mathbf{k{'}}$ is obtained from~$\bm{k}$ through the function~$\LBGamma(\mathbf{k}, v, c)$, where $k_v>0$ and $c\in \aIntInterval{\min(k_v,d)}$. Similarly, for~$\UBMatrixError{g}$ and~$\UBMatrixTrans{g}$, the non-zero entries follow the function~$\UBGamma(\mathbf{k}, v, c)$ with the same conditions. The mappings $\LBGamma$ and $\UBGamma$ are defined as:
    \begin{align}
        &\UBGammaEle(\mathbf{k},0,k_0) = \LBGammaEle(\mathbf{k},0,k_0) = \begin{cases}  
            k_0 + k_1, &\text{ if } l=0, \\
            k_{l+1},  & \text{ otherwise.} 
            \end{cases} , \; \LBGamma(\bm{k}, g, d) = \bm{k} ,\; \\ 
         & \forall (v,c) \notin \{ (0,k_0), (g,d) \}, \;\UBGammaEle(\mathbf{k},v,c)= \LBGammaEle(\mathbf{k},v,c)  =\begin{cases}
            k_l - c, & \text{ if }  l= v,   \\ 
            k_l + c, & \text{ if }  l = v + 1, \\ 
            k_l,  & \text{ otherwise.}
            \end{cases} \\ 
         & \forall g > 1, \;   \UBGammaEle(\mathbf{k},g,d)=  \begin{cases}
            &d,  \text{ if }  l= g,  \; k_g - d,  \text{ if }  l = g-1,\\ 
            &k_0 + k_1,  \text{ if } l = 0, \; k_{l+1}  \text{ otherwise.} 
            \end{cases} 
        \end{align}
For the case where~$g=1$, $\Gamma_{0}^{\text{ub}, g}(\mathbf{k},1,d)= m-d$ and $ \Gamma_{1}^{\text{ub}, g}(\mathbf{k},1,d)= d$.
The components of~$\LBMatrixTrans{g}$ and $\UBMatrixTrans{g}$ satisfy, 
\begin{align}
&p^{\text{lb},g}_{\bm{k},  \LBGamma(\mathbf{k}, v, c)} = p^{\text{ub},g}_{\mathbf{k}, \UBGamma(\mathbf{k}, v, c)} =  \frac{\binom{\sum_{l>v} k_l}{d-c}\binom{k_v}{c}}{\binom{m}{d}}.
\end{align}
The components of~$\LBMatrixError{g}$ and $\UBMatrixError{g}$ satisfy, 
\begin{align}
&\forall (v,c)\neq (g,d), \;\beta_{\mathbf{k},\LBGamma(\mathbf{k},v,c)}^{\text{lb},g} = \beta_{\mathbf{k},\UBGamma(\mathbf{k},v,c)}^{\text{ub},g} =  \frac{\binom{\sum_{w>v} k_w + c}{d} - \binom{\sum_{w>v} k_w}{d}} {\binom{m}{d}}, \\ 
&\beta_{\mathbf{k},\LBGamma(\mathbf{k},g,d)}^{\text{lb},g} = 0,\; \beta_{\mathbf{k},\UBGamma(\mathbf{k},g,d)}^{\text{ub},g} =   \frac{1+ \binom{m}{d}-\binom{m-k_0}{d}}{\binom{m}{d}}.
\end{align}

We numerically compute $\ell_g(T)$ and $U_g(T)$ for $m=50$, $d=4$, and $T=250$. Table~\ref{tab:numericalResults} presents the results for multiple values of $g$, along with the execution time in seconds. 
\begin{table}[htbp]
    \centering
    \caption{Results for different values of $g$ with $m=50$, $d=4$, and $T=250$.}
    \label{tab:results}
    \begin{tabular}{l|l|l|l|l|l}
    \toprule
    $g$ & 1 & 2 & 3 & 4 & 5 \\
    \midrule
    $\ell_g(T)$ & 0.01860 & 0.02956 & 0.03420 & 0.03540 & 0.03559\\ 
    \midrule
    $U_g(T)$  & 0.07654 & 0.04090 & 0.03637 & 0.03572 & 0.03562 \\ 
    \midrule
    Execution time in (s) & 0.05 & 0.68 & 17 & 257 & 3241 \\
    \bottomrule
    \end{tabular}
    \label{tab:numericalResults}
\end{table}

%% file: conclusion.tex
Our paper analyzed $\CMSCU$ with uniform counters selection, a setting we argued constitutes a worst-error scenario. Specifically, in the instance where $d=m-1$, we established closed-form expressions for key metrics within $\CMSCU$. Furthermore, for any given values of $m$ and $d$, we introduced a spectrum of lower and upper bounds on the estimation error, incorporating a parameter~$g$ that enables a flexible balance between computational time and accuracy.  

In future work, we aim to formally prove that $\CMSCU$ with uniform counters selection maximizes the error. Moreover, we plan to identify the appropriate selection of $g$ to attain a desired level of accuracy in error estimation for $\CMSCU$.

%% file: proofs.tex
\section{Useful Lemma}
\label{app:usefulLemma}
For any two vectors~$\mathbf{X}$ and~$\mathbf{Y}$ of equal length. We denote $\mathbf{X}\leq \mathbf{Y}$ to signify that for each index $u$, $X_u \leq Y_u$. The following lemma will aid in proving relevant statements.

\begin{lemma}\label{lem:RecursionDominationCounting}

Let~$(\mathbf{a}(t)= (a_{u}(t))_{u\in \aIntInterval{m}})_{t\in \mathbb{N}}$ and~$(\mathbf{b}(t)= (b_{u}(t))_{u\in \aIntInterval{m}})_{t\in \mathbb{N}}$ be sequences in $\mathbb{N}^{m}$ such that, for any $u\in \aIntInterval{m}$, $a_{u}(t)$ and~$b_u(t)$ are increasing sequences and satisfies: $\mathbf{a}(t+1)\leq \mathbf{a}(t)+ \mathbf{1}$, $\mathbf{b}(t+1)\leq \mathbf{b}(t)+ \mathbf{1}$, and $\mathbf{a}(0)\leq \mathbf{b}(0)$. If the following holds for any $t$ and $u$, 
    \begin{align}\label{e:lem-1}
          \forall j\neq u,\; a_j(t)\leq b_j(t),\;  a_u(t)=b_u(t), \; a_u(t+1) = a_u(t) + 1 \implies  b_u(t+1) = b_u(t) + 1,
    \end{align}
then $\mathbf{a}(t)\leq \mathbf{b}(t)$ for any value of~$t$.
\end{lemma} 
\begin{proof}
    The proof uses recursion. At~$t=0$,~$a_j(0)\leq b_j(0)$, $\forall j\in \aIntInterval{m}$. For $t \geq 1$, we assume that $a_j(t)\leq b_j(t)$ for any $j$. Let~$u\in \aIntInterval{U}$, Whenever~$a_u(t)<b_u(t)$ or~$a_u(t+1) = a_u(t)$, we have $a_u(t+1)\leq b_u(t+1)$. Otherwise, $a_u(t)=b_u(t)$ and $\; a_u(t+1) = a_u(t) + 1$, and using~\eqref{e:lem-1}, we conclude that~$a_u(t+1) \leq b_u(t+1)$, thereby completing the proof.
\end{proof}

\section{Proof of Proposition~\ref{prop:WorstErrorItemNotReq}.}
\label{proof:propositionWorstErrorItem}

    Let~$\mathcal{H}$ be the set of~$d$ hash functions mapping~$\mathcal{I}$ to $\binom{\aIntInterval{m}}{d}$. Let~$i$ and~$x$ be items in~$\mathcal{I}$ such that $n_{x}(t)=0$. Let~$\boldsymbol{\psi}$ be a function from~$\mathcal{H}$ to itself such that for any~$\mathbf{h}= (h_l)_{l\in \aIntInterval{d}}\in \mathcal{H}$, $\mathbf{g}=\boldsymbol{\psi}(\mathbf{h})$ is defined as, 
    \begin{align}
                g_{l}(n) \triangleq \begin{cases}
                    h_l(n),  &\text{ if } n \notin \{i,x \}, \\ 
                    h_l(i), & \text{ if } n = x, \\ 
                    h_l(x), & \text{ otherwise.}
                \end{cases}
    \end{align}
Observe that $\boldsymbol{\psi}$ is a bijection and therefore we can compute the difference between the expected error of items~$i$ and~$x$ as follows, 
    \begin{align}
         \E{e_{i}(T)} - \E{e_{x}(T)} = \frac{1}{\binom{m}{d}} \sum_{\mathbf{h} \in \mathcal{H}} e_{i}(\mathbf{h},t) -e_{x}(\mathbf{g}, T),
    \end{align}
where $e_{i}(\bm{h},t)$ is the error of item $i$ when the hash functions $\bm{h}$ are used. Next, we prove that any~$\mathbf{h} \in \mathcal{H}$ satisfies
        \begin{align}\label{e:prop1-firstStep}
                e_{i}(\mathbf{h}, t) - e_{x}(\mathbf{g}, t) \leq 0. 
        \end{align}
For this purpose, we define a new stochastic process~$(\mathbf{A}(\mathbf{h}, t))_t$ as follows: $\bm{A}(\bm{h},0)=\bm{0}$ and for any $t\geq 0$, 
\begin{align}
       A_{c}(t+1) =
        \begin{cases} 
                    A_c(t) + 1, &\text{if } c \in \argmin_{u\in \bm{h}(r_t)}
                    A_{u}(t), \text{ and } r_t \neq i \\
                    A_{c}(t), & \text{otherwise.}
        \end{cases}
    \end{align}
Observe that $\mathbf{A}(\mathbf{h}, t)=\mathbf{A}(\mathbf{g}, t)$. We prove using recursion that, 
        \begin{align}\label{e:prop1-1}
             Y_u(\mathbf{h},t)  -n_{i}(t)    \leq     A_{u}(\mathbf{h}, t) \leq    Y_u(\mathbf{h},t), \; \forall t. 
        \end{align}
The same result holds when replacing~$\mathbf{h}$ by~$\mathbf{g}$. We deduce form~\eqref{e:prop1-1} that, 
\begin{align}
        &\min_{l\in \aIntInterval{d}} A_{g_l(x)} (\mathbf{g}, t) \leq \min_{l\in \aIntInterval{d}} Y_{g_l(x)} (\mathbf{g}, t) = e_x(\mathbf{g},t), \\
       &\min_{l\in \aIntInterval{d}} A_{h_l(i)} (\mathbf{h}, t) \geq  \min_{l\in \aIntInterval{d}} Y_{h_l(i)} (\mathbf{h}, t) - n_i(t) = e_i(\mathbf{h},t).
\end{align}
By definition of~$\mathbf{g}$, $g_l(x)=h_l(i)$ and therefore~$A_{g_l(x)} (\mathbf{g}, t)=  A_{h_l(i)} (\mathbf{h}, t)$, proving~\eqref{e:prop1-firstStep} and concluding the proof. We use Lemma~\ref{lem:RecursionDominationCounting} in Appendix~\ref{app:usefulLemma} to prove~\eqref{e:prop1-1}. 

\paragraph*{RHS of \eqref{e:prop1-1}.} Suppose that~ $\bm{A}(\bm{h},t)\leq \bm{Y}(\mathbf{h}, t)$, $A_{u}(\mathbf{h},t) = Y_{u}(\mathbf{h}, t)$ and~$A_{u}(\mathbf{h},t+1) = A_{u}(\mathbf{h}, t)+1$. The counter~$u$ was incremented in~$\bm{A}(\bm{h},t)$ means that $\min_{j\in \bm{h}(r_t)} A_{j}(\mathbf{h}, t)=A_{u}(\mathbf{h}, t)$. We can then write, $\min_{j\in \bm{h}(r_t)} Y_{j}(\mathbf{h}, t)\geq \min_{j\in \bm{h}(r_t)} A_{j}(\mathbf{h}, t)=A_{u}(\mathbf{h}, t)=Y_{u}(\mathbf{h}, t)$. Therefore,  $Y_{u}(\mathbf{h}, t+1)=Y_{u}(\mathbf{h}, t)+1$ and Lemma~\ref{lem:RecursionDominationCounting} proves the RHS. 

\paragraph*{LHS of \eqref{e:prop1-1}.} Suppose that~ $\bm{Y}(\bm{h},t)\leq \bm{A}(\bm{h},t)+ \bm{n}_{i}(t)$, $Y_{u}(\mathbf{h}, t)= A_{u}(\mathbf{h},t)  + n_{i}(t)$ and~$Y_{u}(\mathbf{h}, t+1)= Y_{u}(\mathbf{h}, t) + 1$. It follows that $\exists c \in \bm{h}(r_t)$ such that, $Y_{c}(\mathbf{h}, t)\geq Y_{u}(\mathbf{h}, t)= A_{u}(\mathbf{h},t)  + n_{i}(t)$. Moreover, from the recursion hypothesis, we have that~$Y_{c}(\mathbf{h}, t)\leq A_{c}(\mathbf{h},t)  + n_{i}(t)$ and therefore, $A_c(\mathbf{h},t)\geq A_u(\mathbf{h},t)$. Whether $r_t=i$ or not, $A_{u}(\mathbf{h},t+1)  + n_{i}(t+1)=A_{u}(\mathbf{h},t)  + n_{i}(t)+1$, which proves the LHS.

\section{Proof of Theorem~\ref{th:main-particular}}
\label{proof:th-main-particular}

To compute the estimation error of $\CMSCU$ in the regime where~$d=m-1$, we study the stochastic evolution of the number of counters far from the minimum value counter by any value~$l$, denoted~$\Delta_{l}(t)$. Formally,
\begin{align}
        \Delta_l(t) \triangleq \sum_{j=1}^{m} \mathds{1}\left( Y_{j}(t) - Y_{\min}(t) = l   \right).
\end{align}
The stochastic process~$(\boldsymbol{\Delta}(t)=(\Delta_l(t))_{l\in \mathbb{N}})_t$ is discrete-time Markov chain. 

\paragraph*{State space.}  By definition, $\Delta_{0}(t)\geq 1$ and $\sum_{l=0}^{+\infty} \Delta_{l}(t)= m$ at any step~$t$. Moreover, any counter~$j$ attaining the maximum value at step~$t$, i.e., $Y_{j}(t)=Y_{\max}(t)$, is only incremented if it is selected along with counters of the same value. Therefore, $\Delta_{G(t)}(t)\geq d=m-1$, where~$G(t)$ denotes the gap between the maximum and minimum counter values, namely,~$G(t)=Y_{\max}(t) - Y_{\min}(t)$. The state space of $(\Excess(t))_t$, denoted $\Omega$, can be expressed as,
    \begin{align}
        \Omega = \left \{\bm{k} \in \mathbb{N}^{\mathbb{N}}: \sum_{f=0}^{+\infty} k_f = m \text{ and }\; \left(k_0= m \text{ or } \left( k_0 = 1 , k_{L(\bm{k})} = m-1 \right) \right)\right\}.  
    \end{align}
This expression of~$\Omega$ shows that the function $L: \Omega \mapsto \mathbb{N}$, where $L(\bm{k})$ denotes the largest index $f \in \mathbb{N}$ such that $k_f > 0$, is bijective. Consequently, we can model $(\Delta(t))_t$ as a discrete-time birth and death Markov chain on~$\mathbb{N}$~\cite[Sec. 6.3, p. 368]{ross2014introduction}. 

\paragraph*{Transition probabilities.}
Let $p_{L(\bm{k}),L(\bm{k^{'}})}$ be the transition probability from state~$\bm{k}$ to~$\bm{k^{'}}$. Initially, $L(\Delta(0))=0$, and with probability (w.p.)~$1$, $L(\Delta(1))=1$. If $L(\Delta(t))= f \geq 1$, then~$L(\Delta(t+1))= f+1$ only if all maximum value counters are selected, which happens w.p.~$1/m$. Otherwise, $L(\Delta(t+1))= f-1$ w.p.~$m-1/m$. We conclude that~$p_{0,1}=1$, and for any~$i\geq 1$, $p_{i,i+1}=1/m$ and $p_{i,i-1}=m-1/m$.
\paragraph*{Limiting distribution.} 
The Markov chain $(\boldsymbol{\Delta}(t))_t$ has a limiting distribution~\cite[Sec. 6.5, p. 387]{ross2014introduction}. The limiting probability of $\Excess(t)=\bm{k}$, denoted $\pi_{L(\bm{k})}$, is expressed as, 

\begin{align}
  \pi_0 = \frac{m-2}{2(m-1)}, \; \forall f\geq 1, \;  \pi_f = \frac{m(m-2)}{2(m-1)^{f+1}}. 
\end{align}

\paragraph*{Ergodicity.} 

We utilize~\cite[Prop. 4.3, p. 215]{ross2014introduction} to prove the almost sure convergence of quantities of interest to constants. For completeness, we present the proposition. 
\begin{proposition}\cite[Prop. 4.3, p. 215]{ross2014introduction}\label{prop:MarkovErgodicAppendix}
  Let $(X(t))_t$ be an irreducible Markov chain with stationary probabilities $\Pi_f$ for any $f\in\mathbb{N}$, and let $r$ be a bounded function on the state space. Then, 
            \begin{align}
              \lim_{T\to \infty} \frac{1}{T} \sum_{t=1}^{T} r(X_t) = \sum_{f=0}^{\infty} r(f) \Pi_f , \text{ w.p. 1}. 
            \end{align}
\end{proposition}

\noindent \textbf{Error rate.} Let~$\delta_t$ be a Bernoulli random variable, defined as the difference between~$e^{*}(t)$ and~$ e^{*}(t-1)$. We can interpret~$e^{*}(t)$ as the estimation error for an item $j$ such that $n_{j}(t)=0$. We prove that, 
        \begin{align}
        \lim_{T\to \infty} \frac{e^{*}(T) }{T}  = \frac{1}{2}.
        \end{align}
For this purpose, we utilize Proposition~\ref{prop:MarkovErgodicAppendix} with $X(t)=(\delta_t, \Excess(t-1), \Excess(t))$ and $r(X(t))= \delta_t$. We observe that $(X(t))_t$ is an irreducible Markov chain. We denote by~$\Pi_{\delta, L(\bm{k}), L(\bm{k^{'}})}(t)$ the probability that $X(t)$ is at the state $(\delta, \bm{k}, \bm{k^{'}})\in \{0,1\} \times \Omega^{2}$. For any valid values of $f,l\in \mathbb{N}$, this probability can be computed as, 
    \begin{align}
            \Pi_{1, f, l}(t) = \beta_{f, l} \cdot p_{f,l}\cdot \pi_{f}(t). 
    \end{align}
where~$\beta_{f,l}$ is the probability that $\delta_t = 1$ given that $(L(\Excess(t-1)),L(\Excess(t)))= (f,l)$.  
We conclude using Proposition~\ref{prop:MarkovErgodicAppendix} that, 
        \begin{align}
            \lim_{T\to \infty} \frac{e^{*}(T) }{T} 
&= \sum_{\bm{k}, \bm{k^{'}} \in \Omega} \Pi_{1,L(\bm{k}),L(\bm{k^{'}})} 
= \Pi_{1,1,0} + \sum_{f=1}^{\infty} \Pi_{1, f-1 , f}  +  \Pi_{1, f+1 , f} \\
&=  \frac{1}{m} p_{0,1} \pi_0 + \sum_{f=1}^{\infty} \frac{m-1}{m} p_{f,f-1} \pi_f + \frac{1}{m} p_{f,f+1} \pi_f= \frac{1}{2}. 
        \end{align}
\noindent{\textbf{Bounded gap.}} We calculate the fraction of time where the gap is larger than a constant~$g\in \mathbb{N}\setminus \{0\}$. Using Proposition~\ref{prop:MarkovErgodicAppendix} with $X(t)= \Excess(t)$ and $r(X(t))= \mathds{1}(L(\Excess(t)) \geq g)$, we can write, 
    \begin{align}
     \lim_{T\to \infty}  \frac{1}{T}  \sum_{t=1}^{T} \mathds{1}\left( G(t) \geq  g\right)  =\sum_{f=g}^{\infty} \pi_f = \frac{m}{2(m-1)^{g}}. 
    \end{align}
Therefore, the gap is smaller or equal to $1$ with high probability in $m$ and $T$. 




\noindent \textbf{Growth Rate.} We determine the average number of increments per step. Employing Proposition~\ref{prop:MarkovErgodicAppendix}, we establish that this quantity equals $\frac{m}{2}$. Specifically, we set $X(t)= (\Excess(t-1),\Excess(t))$ and  $r(X(t))$ as follows: 
\begin{align*}
r(X(t)) = \begin{cases}
m-1, & \text{if } L(\Excess(t))= L(\Excess(t-1))+1, \\
1, & \text{if } L(\Excess(t))= L(\Excess(t-1))-1, \\
0, & \text{otherwise.}
\end{cases}
\end{align*}
Define $\Pi_{ L(\bm{k}), L(\bm{k^{'}})}(t)$ as the probability that~$X(t)$ is at the state $( \bm{k}, \bm{k^{'}})\in  \Omega^{2}$. Let $C(t)$ be the number of counters incremented at step~$t$. Applying Proposition~\ref{prop:MarkovErgodicAppendix}, we get, 
        \begin{align}
            \lim_{T \to \infty} \frac{1}{T} \sum_{t=1}^{T} C(t) 
&= \sum_{ \bm{k}, \bm{k^{'}}} r( \bm{k}, \bm{k^{'}}) \Pi_{ L(\bm{k}), L(\bm{k^{'}})}  \\ 
&= (m-1)p_{0,1}\pi_0 + \sum_{f=1}^{\infty} (m-1)p_{f,f+1} \pi_f +  (1\cdot p_{f,f+1} \pi_f )\\
&= (m-1) \pi_0 + (1-\pi_0) \frac{2(m-1)}{m}=\frac{m}{2}.
        \end{align}
It follows that the average counter rate is equal to~$\frac{1}{2}$.

\section{Proof of Lemma~\ref{lem:monotony}.}
\label{proof:monothony}

\begin{lemma}\label{lem:LBIncreasing}
$ \LBCBFCU{g}{T} \leq \LBCBFCU{g+1}{T}$.

\end{lemma} 

\begin{lemma}\label{lem:UBDecreasing}
$\UBCBFCU{g+1}{T} \leq \UBCBFCU{g}{T}$. 
\end{lemma} 


Notice that the gap between the maximum and minimum counter values in both LB-CU and UB-CU with the same parameter~$g$ is at most~$g$. This observation can be easily shown by recursion and will help prove Lemmas~\ref{lem:LBIncreasing} and~\ref{lem:UBDecreasing}. 

\begin{proof}[\textbf{Proof of Lemma~\ref{lem:LBIncreasing}}]
    The proof relies on Lemma~\ref{lem:RecursionDominationCounting}. Consequently, for any~$u\in \aIntInterval{m}$,  we first prove that,

    \begin{align}
        \begin{cases}
            \text{(a)}& \LBCBFCU{g}{t}\leq \LBCBFCU{g+1}{t}, \; \LBCBFCUElement{g}{u}{t}= \LBCBFCUElement{g+1}{u}{t}, \\
            \text{(b)}& \LBCBFCUElement{g}{u}{t+1}= \LBCBFCUElement{g}{u}{t}+1,
        \end{cases}
        \implies \LBCBFCUElement{g+1}{u}{t+1}= \LBCBFCUElement{g+1}{u}{t}+1. 
    \end{align}

\noindent We assume (a) and (b). Given (b), it follows that~$u\in\argmin_{j\in \bm{h}(r_t)} \LBCBFCUElement{g}{j}{t}$. Utilizing (a), we can write
    \begin{align}\label{e:lemLBIncre1}
\LBCBFCUElement{g+1}{u}{t}= \LBCBFCUElement{g}{u}{t} = \min_{j\in \bm{h}(r_t)} \LBCBFCUElement{g}{j}{t} \leq \min_{j\in \bm{h}(r_t)} \LBCBFCUElement{g+1}{j}{t} \implies u\in \argmin_{j\in \bm{h}(r_t)} \LBCBFCUElement{g+1}{j}{t}.
    \end{align}
We distinguish three cases. 

\noindent \textbf{(i)} $\LBCBFCUElement{g}{u}{t}< \LBCBFCUElement{g}{\max}{t}$. Utilizing (a) again, we get
\begin{align}\label{e:lemLBIncre2}
    \LBCBFCUElement{g+1}{u}{t}= \LBCBFCUElement{g}{u}{t}< \LBCBFCUElement{g}{\max}{t} \leq \LBCBFCUElement{g+1}{\max}{t} \implies  \LBCBFCUElement{g+1}{u}{t} < \LBCBFCUElement{g+1}{\max}{t}. 
\end{align}
Combining~\eqref{e:lemLBIncre1} and~\eqref{e:lemLBIncre2}, we deduce that~$\LBCBFCUElement{g+1}{u}{t+1}= \LBCBFCUElement{g+1}{u}{t}+1$.

\noindent \textbf{(ii)} If $\LBCBFCUElement{g}{u}{t}= \LBCBFCUElement{g}{\max}{t}$ and $\LBCBFCUElement{g+1}{u}{t}< \LBCBFCUElement{g+1}{\max}{t}$, from~\eqref{e:lemLBIncre1} we deduce that~$\LBCBFCUElement{g+1}{u}{t+1}= \LBCBFCUElement{g+1}{u}{t}+1$.

\noindent \textbf{(iii)} If $\LBCBFCUElement{g}{u}{t}= \LBCBFCUElement{g}{\max}{t}$ and~$\LBCBFCUElement{g+1}{u}{t}= \LBCBFCUElement{g+1}{\max}{t}$, from (b) we infer that~$\LBCBFCUElement{g}{\text{max}}{t} - \LBCBFCUElement{g}{\text{max}}{t} < g$, hence (a) allows us to write

\begin{align}
\LBCBFCUElement{g+1}{\max}{t}- \LBCBFCUElement{g+1}{\min}{t}\leq \LBCBFCUElement{g}{\max}{t} -\LBCBFCUElement{g}{\min}{t} <g \implies \LBCBFCUElement{g+1}{\max}{t}- \LBCBFCUElement{g+1}{\min}{t} < g+1,
\end{align}
thus~\eqref{e:lemLBIncre1} implies that $\LBCBFCUElement{g+1}{u}{t+1}= \LBCBFCUElement{g+1}{u}{t}+1$. We can now conclude the proof by applying Lemma~\ref{lem:RecursionDominationCounting}.

\end{proof}

\begin{proof}[\textbf{Proof of Lemma~\ref{lem:UBDecreasing}}]
      The proof relies on Lemma~\ref{lem:RecursionDominationCounting}. Consequently, for any~$u\in \aIntInterval{m}$,  we first prove that,
    \begin{align}
        \begin{cases}
            \text{(a)}& \UBCBFCU{g+1}{t}\leq \UBCBFCU{g}{t}, \; \UBCBFCUElement{g+1}{u}{t}= \UBCBFCUElement{g}{u}{t}, \\
            \text{(b)}& \UBCBFCUElement{g+1}{u}{t+1}= \UBCBFCUElement{g+1}{u}{t}+1,
        \end{cases}
        \implies \UBCBFCUElement{g}{u}{t+1}= \UBCBFCUElement{g}{u}{t}+1. 
    \end{align}
    We assume (a) and (b). From (b), we deduce that there are two cases.

\noindent \textbf{(i)} $u\in \argmin_{j\in \bm{h}(r_t)} \UBCBFCUElement{g+1}{u}{t}$, and following the same reasoning as in~\eqref{e:lemLBIncre1} from Lemma~\ref{lem:LBIncreasing}, we deduce that~$u\in \argmin_{j\in \bm{h}(r_t)} \UBCBFCUElement{g}{j}{t}$ and therefore~$\UBCBFCUElement{g}{u}{t+1}= \UBCBFCUElement{g}{u}{t}+1$. 

\noindent \textbf{(ii)} We employ a proof by contradiction to establish that~$\UBCBFCUElement{g+1}{u}{t}>\UBCBFCUElement{g+1}{\min}{t}$. Suppose~$\UBCBFCUElement{g+1}{u}{t}=\UBCBFCUElement{g+1}{\min}{t}$, then from (b), we have $\UBCBFCUElement{g+1}{\text{max}}{t} - \UBCBFCUElement{g+1}{u}{t}= g+1$. Furthermore, (a) implies $\UBCBFCUElement{g}{\text{max}}{t} \geq \UBCBFCUElement{g+1}{\text{max}}{t}$ and $\UBCBFCUElement{g+1}{u}{t} =\UBCBFCUElement{g}{u}{t}$. Thus,

\begin{align}
\UBCBFCUElement{g}{\text{max}}{t} - \UBCBFCUElement{g}{u}{t} \geq \UBCBFCUElement{g+1}{\text{max}}{t} - \UBCBFCUElement{g+1}{u}{t} = g+1.
\end{align}
This leads to a contradiction since it is always true that~$\UBCBFCUElement{g}{\text{max}}{t} - \UBCBFCUElement{g}{u}{t}\leq g$. We conclude the proof by applying Lemma~\ref{lem:RecursionDominationCounting}.
    
\end{proof}

\section{Proof of Properties~\ref{eq:complexity} and~\ref{eq:sandwich_proba} in Th.~\ref{th:main}}
\label{proof:dynamicsExcessLBUBprocesses}

\paragraph*{State space.}  By definition, $\LBExcessEle{g}{0}(t)\geq 1$ and $\sum_{l=0}^{+\infty} \LBExcessEle{g}{l}(t)= m$ at any step~$t$. Moreover, any counter~$j$ attaining the maximum value at step~$t$, i.e., $\LBCBFCUElement{g}{j}{t}=\LBCBFCUElement{g}{\max}{t}$, is only incremented if it is selected along with counters of the same value. Therefore, $\LBExcessEle{g}{\LBGap{g}(t)}(t)\geq d$, where~$\LBGap{g}(t)$ denotes the gap between the maximum and minimum counter values, namely,~$\LBGap{g}(t)=\LBCBFCUElement{g}{\max}{t} - \LBCBFCUElement{g}{\min}{t}$. Furthermore, this gap is always smaller than~$g$. It is then easy to show that the state space of~$(\LBExcess{g}(t))_t$, denoted~$\Omega_g$, is given by, 
 \begin{align}\label{e:DefOmegag}
          \Omega_g = \left\{ \mathbf{k} = (k_f)_{f\in \mathbb{N}}\in \mathbb{N}^{\mathbb{N}}: \;  \sum_{f=0}^{+\infty} k_f = M ,\; k_{0}\geq 1, \; k_{L(\mathbf{k})} \geq d, \;  L(\mathbf{k}) \leq g  \right\}, 
\end{align}
where~$L(\mathbf{k}) = \max \left\{ f \in \mathbb{N}: \; k_f > 0 \right\}$. Notice that $\LBGap{g}(t) = L(\LBExcess{g}(t))$. Using a \textit{stars and bars} argument~\cite[p.~38]{feller1968introduction}, we deduce that the cardinality of $\Omega_g$ is equal to~$\binom{m+g-d}{g}$.

To further justify the cardinality of $\Omega_g$, introduce $\Omega_{g,l}$ as the subset of~$\Omega_g$ where for any~$\bm{k}\in \Omega_{g,l}$, $L(\bm{k})=l$. Now, we define $\Omega_{g,l}$ as the subset of~$\Omega_g$ such that for any~$\bm{k}\in \Omega_{g,l}$, $L(\bm{k})=l$. Consider a sequence of integers $i_1 < i_2 < \ldots < i_l$, where each $i_j$ belongs to the interval $\abIntInterval{1}{m+l-d-1}$. Let $k_0 = i_1$, and for each $v \in \abIntInterval{1}{l-1}$, define $k_v = i_{v+1} - i_v - 1$, and $k_l = m+l - 1 - i_l$. This mapping establishes a one-to-one correspondence between the elements of $\Omega_{g,l}$ and the possible sequences $i_1 < i_2 < \ldots < i_l$. Consequently, the cardinality of $\Omega_{g,l}$ is $\binom{m+l-d-1}{l}$. As a result, the cardinality of $\Omega_g$ can be obtained by summing over all possible values of $l$, ranging from $0$ to $g$, giving $\sum_{l=0}^{g} \binom{m+l-d-1}{l}$. This sum is equal to~$\binom{m+g-d}{g}$ (as shown in~\cite[Sec. 5.1, p.159]{ConcreteMathKnuth1994}).

\paragraph*{One-step transition probabilities.} Let~$\LBVMin{g}(t)$ be the difference between the minimum value counters in~$\bm{h}(r_t)$ and~$\LBCBFCUElement{g}{\min}{t}$, and~$\LBNumberVMin{g}(t)$~represents the number of counters in~$\bm{h}(r_t)$ with a value equal to~$\LBVMin{g}(t)$. Formally, 
    \begin{align}
       & \LBVMin{g}(t) \triangleq \min \left\{ \LBCBFCUElement{g}{j}{t} - \LBCBFCUElement{g}{\text{min}}{t}:\;  j\in \bm{h}(r_t)     \right\}, \\  
       & \LBNumberVMin{g}(t) \triangleq  \left |  \left\{ j\in \bm{h}(r_t) :\;  \LBCBFCUElement{g}{j}{t} - \LBCBFCUElement{g}{\text{min}}{t} =\LBVMin{g}(t)   \right\} \right|. 
    \end{align}
Observe that for any $\mathbf{k}\in \Omega_g$, any $v$ satisfying $k_v\geq 1$, and any $c\in \abIntInterval{1}{\min(d,k_v)}$, if $(\LBExcess{g}(t),\LBVMin{g}(t), \LBNumberVMin{g}(t))$ equals $(\mathbf{k}, v ,c)$, then there exists a unique vector~$\mathbf{k^{'}}\in \Omega_g$ such that $\LBExcess{g}(t+1)=\mathbf{k^{'}}$, denoted as $\LBGamma(\mathbf{k}, v, c)$. Therefore, the number of non-zero entries in~$\LBMatrixTrans{g}$ in the row corresponding to~$\mathbf{k}$ is equal to the number of possible values of~$(\LBVMin{g}(t),\LBNumberVMin{g}(t))$ given that $\LBExcess{g}(t)=\mathbf{k}$. This count is, in turn, equal to~$\sum_{v: k_v\geq 1} \min(d,k_v)$, which is at most~$m$, proving that each row of~$\LBMatrixTrans{g}$ has at most~$m$ non zero-entries. Thanks to Assumption~\ref{assum:UniformCounters}, one could compute the transition probabilities as follows,  
\begin{align}\label{e:probaTransLB}
p^{\text{lb},g}_{\mathbf{k}, \LBGamma(\mathbf{k}, v, c)} = \ProbaS{(\LBVMin{g}(t), \LBNumberVMin{g}(t)) = (v,c)}{\LBExcess{g}(t)=\mathbf{k}} = \frac{\binom{k_v}{c} \binom{\sum_{l>v} k_l}{d-c}}{\binom{m}{d}}.
\end{align}
To determine $\LBGamma(\mathbf{k}, v, c)$, we distinguish three cases based on the dynamics of $(\LBCBFCU{g}{t})_t$ (see Alg.~\ref{alg:LB-CBF-CU}). 
\begin{enumerate}
    \item $\LBVMin{g}(t)=g$ and $\LBNumberVMin{g}(t)=d$, no updates are made, and therefore $\LBGamma(\mathbf{k}, g, d)= \mathbf{k}$.

    \item  $\LBVMin{g}(t)=0$ and $\LBNumberVMin{g}(t)=\LBExcessEle{g}{0}(t)$ which implies that~$\LBCBFCUElement{g}{\min}{t+1}=\LBCBFCUElement{g}{\min}{t}+1$, and thus~$\LBExcessEle{g}{0}(t+1)= \LBExcessEle{g}{0}(t)+\LBExcessEle{g}{1}(t)$, and for any $l\geq 1$, $\LBExcessEle{g}{l}(t+1) = \LBExcessEle{g}{l+1}(t)$. It follows that whenever $k_0\leq d$,
    \begin{align}\label{e:LBTransCUVZero}
     \LBGammaEle(\mathbf{k},0,k_0) = \begin{cases}
            k_0 + k_1, &\text{ if } l=0, \\
            k_{l+1},  & \text{ otherwise.} 
            \end{cases}
    \end{align}
    
    \item Otherwise~$\LBCBFCUElement{g}{\min}{t+1}=\LBCBFCUElement{g}{\min}{t}$, and therefore, $\LBExcessEle{g}{\LBVMin{g}(t)}(t+1)= \LBExcessEle{g}{\LBVMin{g}(t)}(t) - \LBNumberVMin{g}(t)$ and~$\LBExcessEle{g}{\LBVMin{g}(t)+1}(t+1)= \LBExcessEle{g}{\LBVMin{g}(t)+1}(t) + \LBNumberVMin{g}(t)$. Hence, for any values of $(v,c)\notin \{(g,d), (0,k_0)\}$, 
    \begin{align}\label{e:LBSTransCU}
          \LBGammaEle(\mathbf{k},v,c)  \triangleq \begin{cases}
            k_l - c, & \text{ if }  l= v,   \\ 
            k_l + c, & \text{ if }  l = v + 1, \\ 
            k_l,  & \text{ otherwise.}
            \end{cases}
    \end{align}

\end{enumerate}

\paragraph*{Limiting distribution} 

The Markov process $(\LBExcess{g}(t))_t$ has a limiting distribution. This is evident due to the existence of a non-zero probability path between any two states $\bm{k}$ and $\bm{k^{'}}$ in $\Omega_g$. Furthermore, starting from the initial state where $k_0=M$ and $k_l=0$ for any $l>0$, it's feasible to return within $m-d+2$ and $m-d+1$ steps, ensuring the chain's aperiodicity. Given the finiteness of $\Omega_g$, we conclude that $(\LBExcess{g}(t))_t$ is \textit{ergodic}, thus possessing a unique limiting distribution irrespective of the initial state.

\paragraph*{Conditional expectation of the error.}  The objective is to compute~$\LBProbaError{g}{k}{k'}$, the probability that $\delta_{t}^{\text{lb},g}=1$ given that $(\LBExcess{g}(t),\LBExcess{g}(t-1)) = (\mathbf{k^{'}},\mathbf{k})$. The quantity~$\LBProbaError{g}{k}{k'}$ is only defined when the transition probability from $\mathbf{k}$ to $\mathbf{k^{'}}$ is strictly positive. Therefore, we assume that~$\mathbf{k^{'}}= \LBGamma(\mathbf{k},v, c)$ such that $k_v>0$ and $c\leq \min(k_v,d)$ (see~\eqref{e:probaTransLB}). The conditional probability that $\delta_{t}^{\text{lb},g}=1$ is equal to the probability that the error for an item~$j$, such that~$n_{j}(t)=0$, increases by~$1$ at step~$t$.

When~$(v,c)=(g,d)$, $\LBCBFCU{g}{t+1}=\LBCBFCU{g}{t}$, and $\mathbf{k^{'}}=\mathbf{k}$ and therefore the error for $j$ does not increase at step $t$, i.e., $\LBProbaError{g}{k}{k}=0$. Otherwise, the error can potentially increase, depending on the value of~$\mathbf{h}(j)=\{h_l(j): \; l\in \aIntInterval{d} \}$. Let $u_1,\ldots, u_c$ be the $c$ counters incremented at step~$t$,  the error for item~$j$ increases by~$1$ at step~$t$ if and only if $\mathbf{h}(j)$ belongs to the set 
\begin{align}
       \left\{ S\in \binom{\aIntInterval{m}}{d}: \; \argmin_{u\in S} \LBCBFCUElement{g}{u}{t} \subset \{u_1, \ldots, u_c \} \right\}, 
\end{align}
which can be written as, 

    \begin{align}\nonumber
        &\left\{ S\in \binom{\aIntInterval{m}}{d}: \; S\subset \{u_1, \ldots, u_c\}\cup \{ u\in \aIntInterval{m}: \;\LBCBFCUElement{g}{u}{t} - \LBCBFCUElement{g}{\min}{t} >  v  \} \} \right\} \setminus \\
        &\left\{ S\in \binom{\aIntInterval{m}}{d}: \; S \subset \{ u\in \aIntInterval{m}: \;\LBCBFCUElement{g}{u}{t} - \LBCBFCUElement{g}{\min}{t} >  v  \}  \right\}.
    \end{align}
Under Assumption~\ref{assum:IdealHash}, $\mathbf{h}(j)$ is a uniform random variable over~$\binom{\aIntInterval{m}}{d}$, and therefore the conditional probabilityof the error can be computed as, 
\begin{align}\label{e:UBCondProbaError}
 &\beta_{\mathbf{k},\LBGamma(\mathbf{k},v,c)}^{\text{lb},g} = \frac{\binom{\sum_{w>v} k_w + c}{d} - \binom{\sum_{w>v} k_w}{d}} {\binom{m}{d}},\; \forall (v,c)\neq (g,d).
\end{align}







\paragraph*{The error in UB-CU.} 
The state space of~$(\UBExcess{g}(t))_t$ is identical to that of~$(\LBExcess{g}(t))_t$, namely $\Omega_g$ defined in~\eqref{e:DefOmegag}. 
The transition from~$\UBExcess{g}(t+1)$ to~$\UBExcess{g}(t)$ is deterministic given~$\UBVMin{g}(t)$ and~$\UBNumberVMin{g}(t)$. More specifically, $\UBExcess{g}(t+1)=\UBGamma(\UBExcess{g}(t),\UBVMin{g}(t),\UBNumberVMin{g}(t))$. When $(\UBVMin{g}(t),\UBNumberVMin{g}(t))\neq (g,d)$, the dynamics of~$\LBExcess{g}(t)$ and~$\UBExcess{g}(t)$ are identical, namely, 
    \begin{align}\label{e:UBTransitionFunction1}
             \UBGamma(\mathbf{k},v,c) = \LBGamma(\mathbf{k},v,c), \forall (v,c)\neq (g,d). 
    \end{align}
On the other hand, if $(\UBVMin{g}(t),\UBNumberVMin{g}(t))\neq (g,d)$, counters attaining the minimum value $\UBCBFCUElement{g}{\min}{t}$ and the selected counters are incremented in $\UBCBFCU{g}{t}$ (see Alg.~\ref{alg:UB-CBF-CU}). In this case, we have that, 
    \begin{align}\label{e:UBTransitionFunction2}
            \forall g> 1, \; &\UBGammaEle(\mathbf{k},g,d)=  \begin{cases}
            d, & \text{ if }  l= g,  \\ 
            k_g - d, & \text{ if }  l = g-1,\\ 
            k_0 + k_1, & \text{ if } l = 0, \\ 
            k_{l+1}, & \text{ otherwise.}
    \end{cases} , \\
      &\UBGammaEle(\mathbf{k},1,d)= \begin{cases}
             m-d, &\text{ if } l=0, \\
             d, &\text{ otherwise.}
    \end{cases}
    \end{align}

The transition probability from $\mathbf{k}$ to $\UBGamma(\mathbf{k}, v, c)$ in~$(\UBExcess{g}(t))_t$ is equal to the transition probability from~$\mathbf{k}$ to~$\LBGamma(\mathbf{k}, v, c)$ in $(\LBExcess{g}(t))_t$ given in~\eqref{e:probaTransLB}. Formally, 
\begin{align}
p^{\text{ub},g}_{\mathbf{k}, \UBGamma(\mathbf{k}, v, c)}= p^{\text{lb},g}_{\mathbf{k}, \LBGamma(\mathbf{k}, v, c)},\;  \forall (\mathbf{k},v,c).
\end{align}

The quantities~$\beta_{\mathbf{k},\UBGamma(\mathbf{k},v,c)}^{\text{ub},g}$ and $\beta_{\mathbf{k},\LBGamma(\mathbf{k},v,c)}^{\text{lb},g}$ are equal whenever~$v<g$ thanks to~\eqref{e:UBTransitionFunction1}. We can then use the expression in~\eqref{e:UBCondProbaError} to compute $\beta_{\mathbf{k},\UBGamma(\mathbf{k},v,c)}^{\text{ub},g}$. However, when~$v=g$, all the $k_0$ counters equal to~$\UBCBFCUElement{g}{\min}{t}$ are incremented along with~$d$ counters equal to~$\UBCBFCUElement{g}{\max}{t}$, denoted $u_1,\ldots, u_d$. Therefore we distinguish two cases where the error for an item~$j$, with $n_j(t)=0$, increases by~$1$. The first case is when $\mathbf{h}(i)=\{u_1,\ldots, u_d\}$. The second case is when at least one of the $k_0$ minimal value counters belongs to $\mathbf{h}(i)$. We deduce that, 
\begin{align}
 &  \beta_{\mathbf{k},\UBGamma(\mathbf{k},g,d)}^{\text{ub},g} =  \frac{1+ \binom{m}{d}-\binom{m-k_0}{d}}{\binom{m}{d}}.
\end{align}


